\documentclass[prl,twocolumn,showpacs,preprintnumbers,amsmath,amssymb,superscriptaddress]{revtex4}

\usepackage{bm}
\usepackage{graphicx}
\usepackage{amsbsy}
\usepackage{amsmath}
\usepackage{amsfonts}
\usepackage{amsthm}

\begin{document}

\theoremstyle{plain}
\newtheorem{theorem}{Theorem}
\newtheorem{lemma}[theorem]{Lemma}
\newtheorem{corollary}[theorem]{Corollary}
\newtheorem{conjecture}[theorem]{Conjecture}
\newtheorem{proposition}[theorem]{Proposition}

\theoremstyle{definition}
\newtheorem{definition}{Definition}

\theoremstyle{remark}
\newtheorem*{remark}{Remark}
\newtheorem{example}{Example}

\title{Evolution and Symmetry of Multipartite Entanglement}   

\author{Gilad Gour}\email{gour@math.ucalgary.ca}
\affiliation{Institute for Quantum Information Science and 
Department of Mathematics and Statistics,
University of Calgary, 2500 University Drive NW,
Calgary, Alberta, Canada T2N 1N4}

\date{\today}

\begin{abstract} 
We discover a simple factorization law describing how multipartite entanglement of a composite quantum system 
evolves when one of the subsystems undergoes an arbitrary physical process. This multipartite entanglement decay is 
determined uniquely by a single factor we call the entanglement resilience factor (ERF). Since the ERF is a function of
the quantum channel alone, we find that multipartite entanglement evolves in exactly the same way as bipartite (two qudits)
entanglement. For the two qubits case, our factorization law reduces to the main result of Nature Physics \textbf{4}, 99 (2008).
In addition, for a permutation $P$, we provide an operational definition of $P$-asymmetry of entanglement, and find 
the conditions when a permuted version of a state can be achieved by local means.
\end{abstract}  


\maketitle

With the emergence of quantum information science in recent years, much effort has been given to the study of entanglement~\cite{Hor09,Ple07}. It was realized that highly entangled states are the most desirable resources for a variety of quantum information processing (QIP) tasks, such as quantum teleportation~\cite{Ben93}, super-dense coding~\cite{Ben92}, entanglement-based quantum cryptography~\cite{Eke91}, error
correcting codes~\cite{Sch99}, and more recently, one-way quantum computation~\cite{Bri01}. Due to the effect of decoherance induced by the coupling of the subsystems with the environment, the entanglement of the composite quantum system decreases in time. It is therefore critical, for the implementations of many important QIP tasks, to understand the behaviour of entanglement under the influence of decoherance or noise.

To study the evolution of entanglement it seems to be necessary first to study the evolution of the quantum state describing the composite system and then to calculate its entanglement. For example, a situation where no energy is exchanged with the environment, the master equation involving the Lindbland operators can be used to determine the state evolution. Indeed, the elaborate theory on state evolution was the method used by many researchers (e.g. see references in~\cite{Kon08}). However, the drawback of this technique is that for multipartite systems (or higher dimensional systems) the state equation can be very hard to solve and therefore the evolution of entanglement can be determined only in very special cases. 
Quite recently, a new way was found~\cite{Kon08} (see also~\cite{GS} for an earlier similar work) to characterize the evolution of entanglement in two qubits systems, by which the evolution of concurrence~\cite{Woo98} (a two qubit measure of entanglement) is determined directly in terms of the evolution
of a maximally entangled state; i.e. a Bell state. This technique was generalized to determine the evolution of the G-concurrence~\cite{Gou05} of two qudits in~\cite{Mar08}. In both~\cite{Kon08} and~\cite{Mar08} the authors used the Choi-Jamiolkowski isomorphism in order to derive the equations
describing the time evolution of entanglement. Hence, since the Choi-Jamiolkowski isomorphism applies only for bipartite systems, it may give the impression that such entanglement-evolution equations can
not be extended to multipartite settings.  

In this Letter we discover a simple factorization law describing how multipartite entanglement of a composite quantum system 
evolves when one of the subsystems undergoes an arbitrary physical process. Quite remarkable, this factorization law
holds for arbitrary number of parties, and reduces to the factorization law given in~\cite{Kon08,Mar08} for the bipartite case. 
Our key idea is to use measures of entanglement that are invariant under the group 
$G\equiv \text{SL}(d_1,\mathbb{C})\otimes\text{SL}(d_2,\mathbb{C})\otimes\cdots\otimes\text{SL}(d_n,\mathbb{C})$,
where $d_1,d_2,...,d_n$ are the dimensions of the $n$-subsytems, and $\text{SL}(d,\mathbb{C})$ is the group of $d\times d$ complex matrices with determinant 1. The group $G$ represents (determinant 1) stochastic 
operations assisted by classical communications (SLOCC) and has been used extensively in the classifications of multi-partite entanglement. It is therefore clear from our analysis that even in the bipartite case, it is the invariance under $G$, rather than the Cho-Jemiolkowski isomorphism, that is necessary for the derivation of the factorization law. 

In addition to the factorization law, we also provide an operational definition of $P$-asymmetry of entanglement:
a multipartite entangled state contains $P$-asymmetric entanglement if its subsystems can not be permuted (according to the 
permutation $P$) by means of LOCC. We show that in general states have $P$-asymmetric entanglement, and by using
measures of entanglement that are invariant under $G$, we are able to generalize the main result of~\cite{Hor10} to the case
of multi-partite systems.

While two-party entanglement was very well studied, entanglement in multi-party systems is far less understood.
Perhaps one of the reasons is that $n$ qubits (with $n>3$) can be entangled in an uncountable number of 
ways~\cite{Vid00,Ver03,GW} with respect to SLOCC. 
It is therefore
not very clear what role entanglement monotones can play in multi-qubits system unless they are defined operationally.
One exception from this conclusion are entanglement monotones that are defined in terms of SL-invariant 
polynomials~\cite{GW,Ver03,Uh00,Luq03,Jens,W,Miy03,Cof00}.   

For example, for two qubits the concurrence~\cite{Woo98} is an entanglement monotone that is invariant under the action
of the group $\text{SL}(2,\mathbb{C})\otimes \text{SL}(2,\mathbb{C})$. That is, for a given two qubits state $|\psi\rangle$
the concurrence of $|\psi\rangle$ is the same as the concurrence of the unnormalized state $A\otimes B|\psi\rangle$, where
$A$ and $B$ are $2\times 2$ complex matrices with determinant 1. The concurrence originally was defined as a mathematical 
tool to calculate an operational measure of entanglement, namely, the entanglement of formation. Nevertheless, 
the concurrence, as being the only $\text{SL}(2,\mathbb{C})\otimes \text{SL}(2,\mathbb{C})$ invariant 2-qubits homogeneous measure of degree 1, was used as the key measure of entanglement in numerous QIP tasks (e.g. see~\cite{Hor09,Ple07} and references therein). 

Another example is the square root of the 3-tangle (SRT)~\cite{Cof00}. The SRT is the only $\text{SL}(2,\mathbb{C})\otimes \text{SL}(2,\mathbb{C})\otimes\text{SL}(2,\mathbb{C})$ invariant measure of entanglement that is homogenous of degree 1. 
The SRT also plays an important role in relation to monogamy of entanglement, and in fact it is the \emph{only} measure that capture the 3-way entanglement.
For 4-qubits or more, the picture is different since there are many homogenous SL-invariant measures of entanglement, such as
the the square root 4-tangle~\cite{Uh00} or the 24th root of the Hyperdeterminant~\cite{Miy03}. We now define all such measures that will be satisfying our factorization law; such measures were first discussed in~\cite{Ver03}.

\begin{definition}\label{ginv}
Set $G\equiv \text{SL}(d_1,\mathbb{C})\otimes\text{SL}(d_2,\mathbb{C})\otimes\cdots\otimes\text{SL}(d_n,\mathbb{C})$,
$\mathcal{H}_n\equiv \mathbb{C}^{d_1}\otimes\mathbb{C}^{d_2}\otimes\cdots\otimes\mathbb{C}^{d_n}$, and 
$\mathcal{B}(\mathcal{H}_n)$ the set of all bounded operators (e.g. density matrices) acting on $\mathcal{H}_n$.
A SL-invariant multi-partite measure of entanglement, $E_{inv}$, is a non-zero function from $\mathcal{B}(\mathcal{H}_n)$
to the non-negative real numbers satisfying the following:\\
\textbf{(i)} It is $G$-invariant; that is, 
$$
E_{inv}(g\rho g^{\dag})=E_{inv}(\rho)\;,
$$
for all $g\in G$ and $\rho\in\mathcal{B}(\mathcal{H}_n)$.\\
\textbf{(ii)} It is homogenous of degree 1; i.e. $$E_{inv}(r\rho)=rE_{inv}(\rho)$$
for all non-negative $r$ and all $\rho\in\mathcal{B}(\mathcal{H}_n)$.\\
\textbf{(iii)} On mixed states it is given in terms of the convex roof extension; That is,
$$
E_{inv}(\rho)=\min\sum_{i}p_iE_{inv}(\psi_i)\;,
$$
where the minimum is taken with respect to all pure states decompositions
of $\rho=\sum_{i}p_i\psi_i$ (here $\psi_i\equiv|\psi_i\rangle\langle\psi_{i}|$ is a rank 1 density matrix).
\end{definition}

\begin{remark}
The criteria in the definition above guarantee that $E_{inv}$ is  
an entanglement monotone~\cite{Ver03}. Note also that the construction via the convex roof extension is consistent
with conditions (i) and (ii).
The concurrence, the G-concurrence, and the SRT are all
satisfying the conditions in the definition above.
It can be easily checked that $E_{inv}$ is unique (up to multiplication by a positive constant) for the bipartite case with $d_1=d_2$ and for three qubits, but it is not unique for $n$-qubits with $n>3$. Indeed, for 4 qubits there are 4 algebraically independent SL-invariant polynomials that generate a whole family of such SL-invariant measures~\cite{Luq03,Jens,W,GW}. Note however that
if the dimensions of the subsystems $\{d_i\}$ are not all equal then a SL-invariant measure $E_{inv}$ may not exists. For example, in the bipartite case with $d_1\neq d_2$, $E_{inv}$ does not exists. For three parties, on the other hand, with $d_1=d_2=2$ and $d_3=3$, such a measure exists; it is given in terms of the Hyperdeterminant~\cite{Miy03}.
\end{remark}

The criteria in the definition above are motivated from two observations. First,
if two states are connected by SLOCC reversible transformation, then essentially they
both must have the same type (though not necessarily the same amount) of multi-qubit entanglement.
That is, let $|\psi\rangle\in\mathcal{H}_n$ and let $|\varphi\rangle=g|\psi\rangle/\|g|\psi\rangle\|$, where $g\in G$. 
Then, criteria (i) and (ii) guarantee that if $E_{inv}(|\psi\rangle)\neq 0$ then also $E_{inv}(|\varphi\rangle)\neq 0$. The second observation says something about
the \emph{amounts} of the multipartite entanglement in $|\psi\rangle$ and $|\varphi\rangle$. It is based on some results first 
discovered in~\cite{Ver03} and discussed further in~\cite{GW}. We now describe shortly these results.

Let $|\psi\rangle\in\mathcal{H}_n$ and consider the set of (in general non-normalized) states $G|\psi\rangle$ (i.e. the orbit
of $|\psi\rangle$ under $G$). By definition, if $|\psi\rangle$ is generic then $G|\psi\rangle$ is closed. 
Therefore, for most states $G|\psi\rangle$ is closed. 
If $G|\psi\rangle$ is not closed then consider its closure, and denote by $|\tilde\varphi\rangle$ the state in $G|\psi\rangle$ with the minimum norm; that is 
$\langle\tilde{\varphi}|\tilde{\varphi}\rangle\leq \langle\tilde{\phi}|\tilde{\phi}\rangle$ for all $|\tilde{\phi}\rangle\in G|\psi\rangle$.
The state $|\varphi\rangle\equiv |\tilde{\varphi}\rangle/\sqrt{\langle\tilde{\varphi}|\tilde{\varphi}\rangle}$ 
is called a \emph{normal form}~\cite{Ver03} (see also the critical set in appendix A of~\cite{GW}). Moreover, note that if
the normalized state $|\psi\rangle$ is a normal form, then $\|g|\psi\rangle\|\geq 1$ for all $g\in G$, with equality if and only if 
$g\in\text{SU}(d_1)\otimes \text{SU}(d_2)\otimes\cdots\otimes\text{SU}(d_n)$~\cite{Ver03,GW}.

The properties (i) and (ii) in the definition above imply that if $|\psi\rangle$ is a normal
form, and if $|\varphi\rangle=g|\psi\rangle/\|g|\psi\rangle\|$, then 
$$
E_{inv}(|\varphi\rangle)= \frac{E_{inv}(g|\psi\rangle)}{\|g|\psi\rangle\|}=\frac{E_{inv}(|\psi\rangle)}{\|g|\psi\rangle\|}\leq E_{inv}(|\psi\rangle)\;.
$$
That is, criteria (i) and (ii) imply that among all the states that can be obtained from $|\psi\rangle$ by SLOCC, the normal form $|\psi\rangle$ has the maximum amount of $E_{inv}$.  Indeed, in~\cite{Ver03,GW} it has been shown that $|\psi\rangle$ is a normal form if and only if each qudit is maximally entangled with the rest of the qudits (i.e. the local density matrices of all qudits are proportional to the identity). Therefore, criteria (i) and (ii) are consistent with this result, and we can consider the normal forms as 
maximally entangled states. 
More details and further motivation for the
first criterion can be found in the extensive literature on the characterization of entanglement in terms of SL-invariant polynomials
(see for example~\cite{GW,Ver03,Uh00,Luq03,Jens,W,Miy03} and references therein).
Our last remark on Def.~\ref{ginv} we summarize in the following lemma.
\begin{lemma}\label{lemma}
Let $|\psi\rangle\in\mathcal{H}_n$ and $E_{inv}$ as defined in Def.~\ref{ginv}.
Then, for a matrix $M:\mathcal{H}_n\to\mathcal{H}_n$ of the form $M=A_1\otimes A_2\otimes\cdots\otimes A_n$
we have $E_{inv}(M|\psi\rangle\langle\psi|M^{\dag})=0$ if there exists $1\leq k\leq n$ such that $\det A_k=0$.
\end{lemma}

\begin{proof}
Without loss of generality, assume $\det A_1=0$. Denote by $|0\rangle$ a normalized vector in $\mathbb{C}^{d_1}$
such that $A_1|0\rangle=0$. Denote by $|k\rangle$ (with $k=1,2,...,d_1-1$) other vectors in $\mathbb{C}^{d_1}$, completing
$|0\rangle$ to a basis. With this basis we can write 
\begin{equation}\label{zero}
M|\psi\rangle=\sum_{k=1}^{d_1-1}|v_k\rangle|\varphi_k\rangle\;,\;\text{with}\;\; 
|\varphi_k\rangle\in\mathbb{C}^{d_2}\otimes\cdots\otimes\mathbb{C}^{d_n},
\end{equation}
and $|v_k\rangle\equiv A_1|k\rangle$.
Note that $|v_k\rangle$ and $|\varphi_k\rangle$ are not necessarily normalized and that the sum above starts from $k=1$ 
and not from $k=0$ since $A_1|0\rangle=0$.
Denote by $P_r$ the projection to the space $V=\text{span}\{|v_k\rangle\}_{k=1}^{d_1-1}$, where $r=\dim V$.
For $0<t\in\mathbb{R}$, let $D_t$ be the following $d_1\times d_1$ matrix:
$$
D_t=t^{r/(r-d_1)}\left(\mathbb{I}-P_r\right)+tP_r\;.
$$
Denote also $g_t\equiv D_t\otimes\mathbb{I}\otimes\cdots\otimes\mathbb{I}$. Clearly, $\det g_t=1$; i.e. $g_t\in G$.
Hence,
\begin{align*}
& E_{inv}\left[M|\psi\rangle\langle\psi|M^{\dag}\right] =E_{inv}\left[g_{t}M|\psi\rangle\langle\psi|M^{\dag}g^{\dag}_{t}\right]\\
& =E_{inv}\left[t^2 M|\psi\rangle\langle\psi|M^{\dag}\right]=t^2E_{inv}\left[M|\psi\rangle\langle\psi|M^{\dag}\right],
\end{align*}
where we have used Eq.(\ref{zero}) and criteria (i) and (ii) of Def.~\ref{ginv}. Since the above equality is true for all $t>0$, we must have $E_{inv}\left[M|\psi\rangle\langle\psi|M^{\dag}\right]=0$.
\end{proof}
We are now ready to discuss the main result of this paper.

\begin{definition}
Let $\$: \mathcal{B}(\mathbb{C}^{d})\to\mathcal{B}(\mathbb{C}^{d})$, be a quantum channel acting on $d\times d$ positive 
semi-definite matrices (i.e. density matrices). Any such channel has Kraus representation 
$\$(\cdot)=\sum_{j}K_j(\cdot)K_{j}^{\dag}$, with Kraus operators $\sum_{j}K_{j}^{\dag}K_{j}\leq 1$. 
We define the \emph{entanglement resilience factor} of $\$$ to be
\begin{equation}\label{erf}
\mathcal{F}[\$]\equiv\min\sum_{j}\left|\det K_j\right|^{2/d},
\end{equation}
where the minimum is taken with respect to all the Kraus representations of $\$$.
\end{definition}
Note that $0\leq \mathcal{F}[\$]\leq 1$ due to the geometric-arithmetic inequality and the fact that 
for all Kraus representations of $\$$, $\sum_{j}K_{j}^{\dag}K_{j}\leq 1$. Recall also that all Kraus representations of
a quantum channel are related by a unitary matrix.
In the theorem below we give an operational interpretation
for $\mathcal{F}[\$]$.
\begin{theorem}\label{maint}
Let $|\psi\rangle\in\mathcal{H}_n$ and $|\phi\rangle\in\mathcal{H}_n$ be two states with non zero value of $E_{inv}$.
Denote by $\Lambda\equiv \$\otimes\mathbb{I}\otimes\cdots\otimes\mathbb{I}$, 
where $\$$ is an arbitrary quantum channel, which may represent the influence of the environment on the first qudit.
Then,
\begin{equation}\label{main}
\frac{E_{inv}\left[\Lambda\left(|\psi\rangle\langle\psi|\right)\right]}
{E_{inv}(|\psi\rangle\langle \mathcal{\psi}|)}=
\frac{E_{inv}\left[\Lambda\left(|\mathcal{\phi}\rangle\langle \mathcal{\phi}|\right)\right]}
{E_{inv}(|\mathcal{\phi}\rangle\langle \mathcal{\phi}|)}=\mathcal{F}[\$]\;.
\end{equation}
That is, the ratio between the final entanglement and the initial entanglement depends solely on the entanglement resilience 
factor of the channel.
\end{theorem}
\begin{remark}
In the bipartite case with $d_1=d_2$ the formula above reduce to the one given in~\cite{Kon08,Mar08}, by replacing the state 
$|\psi\rangle$
with a maximally entangled state, and by taking $E_{inv}$ to be the concurrence~\cite{Woo98} or G-concurrence~\cite{Gou05} for two qubits or two qudits, respectively. However, in~\cite{Kon08,Mar08} the entanglement resilience factor (ERF) of the channel was not introduced. A remarkable observation is that the ERF depends only on a \emph{single} qudit channel. Now, consider the case of two qubits. Then, by taking $E_{inv}$ to be the concurrence $C$ and replacing $|\psi\rangle$ above with a Bell state $|\psi^+\rangle$ we get the following formula for the ERF:
\begin{equation}\label{conERF}
\mathcal{F}[\$]=C\left(\$\otimes\mathbb{I}(|\psi^{+}\rangle\langle\psi^{+}|)\right)\;,
\end{equation}
which can be determined completely by using the Wootters formula. It is remarkable that for any number of qubits and for any choice
of $E_{inv}$, this is the unique formula that is needed to be calculated in order to determine the ERF of a channel acting on a qubit.
Similarly, for a channel $\$$ acting on a qudit, the ERF is given in terms of the G-concurrence~\cite{Gou05}:
\begin{equation}\label{gconERF}
\mathcal{F}[\$]=G\left(\$\otimes\mathbb{I}(|\psi^{+}\rangle\langle\psi^{+}|)\right)\;,
\end{equation} 
where here $|\psi^{+}\rangle$ stands for a maximally entangled state in $\mathbb{C}^d\otimes\mathbb{C}^d$.
Eq.(\ref{gconERF}) provide the unique value of the ERF and can be used to the determine the ratios in Eq.(\ref{main})
independent on the choice of $E_{inv}$ or the number of qudits involved.
\end{remark}
\begin{proof}
Given the channel $\$(\cdot)=\sum_{j}K_j(\cdot)K_{j}^{\dag}$, the density matrix 
$\rho\equiv \left(\$\otimes\mathbb{I}\otimes\cdots\otimes\mathbb{I}\right)|\psi\rangle\langle\psi|$ 
have the following pure state decomposition:
$$
\rho=\sum_{j}|\tilde{v}_j\rangle\langle\tilde{v}_j|\;\;,\;\;|\tilde{v}_j\rangle\equiv K_j\otimes \mathbb{I}\otimes\cdots\otimes \mathbb{I} |\psi\rangle\;.
$$
Note that $|\tilde{v}_j\rangle$ are \emph{not} normalized. Moreover, denote by $\rho=\sum_{i}|\tilde{w}_i\rangle\langle\tilde{w}_i|$  the optimal decompositions of $\rho$.  That is, denote $p_i\equiv\langle \tilde{w}_i|\tilde{w}_i\rangle$ 
(so that $p_{i}^{-1/2}|w_i\rangle$ is normalized) 
$$
E_{inv}(\rho)=\sum_{i}p_iE_{inv}\left(\frac{1}{p_i}|\tilde{w}_i\rangle\langle\tilde{w}_i|\right)=\sum_{i}E_{inv}\left(|\tilde{w}_i\rangle\langle\tilde{w}_i|\right)
$$
where we have used the fact that $E_{inv}$ is homogeneous of degree 1. 
Now, since $\{|\tilde{v}_j\rangle\}$ and $\{|\tilde{w}_i\rangle\}$ 
are two deferent decompositions of $\rho$, they are related to each other via a unitary matrix $U$.
That is, if the two sets $\{|\tilde{v}_j\rangle\}$ and $\{|\tilde{w}_i\rangle\}$ do not have the same number
of vectors we add zero vectors to the smaller set and then we have
$$
|\tilde{w}_i\rangle=\sum_{j}U_{ij}|\tilde{v}_j\rangle=M_i\otimes \mathbb{I}\otimes\cdots\otimes \mathbb{I} |\psi\rangle\;,
$$
where $U$ is a unitary matrix, and $M_i\equiv\sum_{j}U_{ij}K_{j}$ form another Kraus representation to the \emph{same}
quantum channel $\$$. Now, w.l.o.g. (see lemma~\ref{lemma}) we can assume that $\det M_i\neq 0$. Hence, we can write
$$
|\tilde{w}_i\rangle=(\det M_i)^{1/d}\left(\frac{M_i}{(\det M_i)^{1/d}}\otimes \mathbb{I}\otimes\cdots\otimes \mathbb{I}\right) |\psi\rangle\;.
$$
Since $E_{inv}$ is $G$-invariant and homogeneous, we get $E_{inv}(|\tilde{w}_i\rangle)=\left |\det M_i\right|^{2/d}E_{inv}(|\psi\rangle)$ and thus
\begin{equation}\label{sig}
E_{inv}(\rho)=\sum_{i}\left |\det M_i\right|^{2/d}E_{inv}(|\psi\rangle)\;.
\end{equation}
What is left to show is that 
\begin{equation}\label{key}
\mathcal{F}[\$]=\sum_{i}\left |\det M_i\right|^{2/d}\;.
\end{equation}
To see that, note that the unitary $U$ has been chosen such
that the decomposition $\rho=\sum_{i}|\tilde{w}_i\rangle\langle\tilde{w}_i|$ is optimal. Based on Eq.(\ref{sig}), 
$U$ has been chosen such that $\sum_{i}\left |\det M_i\right|^{2/d}$ gets the minimum possible value among all the different
Kraus representations of $\$$. Hence, the equality in Eq.(\ref{key}) must hold. 
\end{proof}
The measure $E_{inv}$ is a convex function since it is defined in terms of the convex roof extension. Thus, we have the following
corollary.
\begin{corollary}
Let $\rho\in\mathcal{H}_n$ be a multipartite mixed state with non-zero value of $E_{inv}$, and let $\Lambda$ and $\$$ be as 
in the theorem above. Then,
$$
\frac{E_{inv}\left[\Lambda\left(\rho\right)\right]}
{E_{inv}(\rho)}\leq\mathcal{F}[\$]\;.
$$ 
\end{corollary}
The following corollary is an immediate consequence of the equation above.
\begin{corollary}
Let $\rho\in\mathcal{H}_n$ be a multipartite mixed state with non-zero value of $E_{inv}$, and let $\{\$_{k}\}_{k=1,2,...,n}$ be a set of $n$ quantum channels. Then,
\begin{equation}\label{nway}
\frac{E_{inv}\left[\$_1\otimes\$_2\otimes\cdots\otimes\$_n\left(\rho\right)\right]}
{E_{inv}(\rho)}\leq\prod_{k=1}^{n}\mathcal{F}[\$_k]\;.
\end{equation}
\end{corollary}

As a simple illustration of the above theorem and corollaries, consider the case of three qubits. In three qubits the only 
$G$-invariant measure of entanglement, $E_{inv}$, is given by the SRT~\cite{Cof00} on pure states,
and on mixed states it is defined in terms of the convex roof extension. The GHZ state maximize this measure.
Applying the theorem above to this measure gives
$$
E_{inv}\big(\$\otimes\mathbb{I}\otimes\mathbb{I}|GHZ\rangle\langle GHZ |\big)=\mathcal{F}[\$]\;,
$$
where $\mathcal{F}[\$]$ can be calculated via Wootters formula (see Eq.(\ref{conERF})).
That is, we have found a closed formula for the SRT for all mixed states of the form
$
\rho=\$\otimes\mathbb{I}\otimes\mathbb{I} |GHZ\rangle\langle GHZ |
$. Moreover, the corollaries above provides upper bounds for $E_{inv}$ on states of the form 
$
\rho=\$_1\otimes\$_2\otimes\$_3 |GHZ\rangle\langle GHZ |
$.

One can also ask how the multipartite entanglement evolve after  a separable measurement is performed by the $n$ parties.
In the following lemma we obtain an upper bound on the ratio between the initial entanglement and the final average entanglement
after such a measurement is performed. This lemma will be very useful in our discussion on symmetry of multipartite entanglement.

\begin{lemma}\label{lem}
Let $\Lambda(\cdot)=\sum_k M_k (\cdot) M_{k}^{\dag}$ be a trace preserving separable operation,
where $M_k\equiv A_{1}^{(k)}\otimes A_{2}^{(k)}\otimes\cdots\otimes A_{n}^{(k)}$ and $\sum_{k}M_{k}^{\dag}M_{k}\leq 1$. 
Then,
\begin{equation}\label{l1}
\frac{\sum_{k}p_kE_{inv}(\sigma_k)}{E_{inv}(\rho)}\leq\sum_{k}|\det M_k|^{2/d}
\end{equation}
where $\sigma_k\equiv \frac{1}{p_k}M_k\rho M_{k}^{\dag}$ and $p_k\equiv\text{Tr}M_k\rho M_{k}^{\dag}$. 
Further,
\begin{equation}
\sum_{k} |\det M_k|^{\frac{2}{d}}  
=\sum_{k}|\det A_{1}^{(k)}|^{\frac{2}{d}}\cdots|\det A_{n}^{(k)}|^{\frac{2}{d}}\leq 1
\label{l2}
\end{equation}
with equality if and only if all the operators $\{A_{i}^{(k)}\}$ are proportional to unitaries. That is, if $\sum_{k}|\det M_k|^{2/d}= 1$ then $\Lambda$ is a mixture of product unitary operations.
\end{lemma}
\begin{proof}
The upper bound in Eq.(\ref{l1}) is a direct consequence of conditions (i) and (ii) in~Def.~\ref{ginv}, and the upper bound in 
Eq.(\ref{l2}) is a direct consequence of the geometric-arithmetic inequality.
\end{proof}

We now show that SL-invariant measures of multipartite entanglement can also be very useful to determine
the symmetry of multipartite entanglement.

\begin{definition}
Let $P$ be a permutation on $n$ parties. Let $\rho: \mathcal{H}_n\to\mathcal{H}_n$ be a multipartite density matrix and denote by 
$V_P\rho V_{P}^{\dag}$ the ``permuted version" of $\rho$, where $V_P$ is unitary operator that permutes the subsystems.
Then, the entanglement contained in a multipartite state, $\rho$, is said to be P-Symmetric, if by LOCC we can produce
the permuted version of the state; i.e. $V_P\rho V_{P}^{\dag}$. 
\end{definition}

The theorem below generalize the main result of~\cite{Hor10} to multipartite states. 

\begin{theorem}\label{sym}
Let $E_{inv}$ in Def.~\ref{ginv} be also invariant under some permutation $P$ of the $n$ qudits.
Let $\rho :\mathcal{H}_n\to\mathcal{H}_n$ be a mixed state for which $E_{inv}(\rho)>0$
and assume that the entanglement of $\rho$ is P-Symmetric. Then, the permuted version 
of the state can be achieved by some product unitary operation
$U_{A_1}\otimes U_{A_2}\otimes\cdots\otimes U_{A_n}$.
\end{theorem}

Note in particular that if a state with $E_{inv}>0$ has different entropies of subsystems, it can not be permuted by LOCC,
as local unitaries can not change local entropy. However, if a state has distinct local entropies, but $E_{inv}=0$,
then it may be possible to permute the state by LOCC. As a simplest example consider the state 
$\rho=\rho_{A_1}\otimes\rho_{A_2}\otimes\cdots\otimes\rho_{A_n}$. Clearly, any such state, or its permuted version, can be generated locally even though its local entropies can be distinct. Note also that
in three qubits, the SRT is also invariant under permutations~\cite{Cof00},
and therefore can be used for the theorem above for all permutations $P$. Same is true for all $E_{inv}$ that are also permutation invariant, such as the 4-tangle~\cite{Uh00} and all Hyperdeterminants~\cite{Miy03} (see also~\cite{Jens} for other such $E_{inv}$).
The proof below of theorem~\ref{sym} is based on lemma~\ref{lem} 
and follows the exact same lines as in Theorem~1 of~\cite{Hor10}.

\begin{proof}
Suppose $E(\rho)>0$ and let $\Lambda$ (as defined in lemma~\ref{lem}) represents an LOCC map that permutes $\rho$; i.e. $\Lambda(\rho)=V_{P}\rho V_{P}^{\dag}$. 
Hence, $\sum_{k}p_k\sigma_k=V_{P}\rho V_{P}^{\dag}$. Using the invariant of $E_{inv}$ under the permutation $P$ we get
\begin{align*}
E_{inv}(\rho)& =E_{inv}(V_{P}\rho V_{P}^{\dag})=E_{inv}(\sum_{k}p_k\sigma_k)\\
&\leq \sum_{k}p_kE_{inv}(\sigma_k)
\leq\sum_{k}|\det M_k|^{2/d}E_{inv}(\rho)
\end{align*}
Since we assumed that $E_{inv}(\rho)>0$ we get that $\sum_{k}|\det M_k|^{2/d}\geq 1$. Thus, from lemma~\ref{lem} we conclude that $\Lambda$ must be a mixture of product unitaries:
$$
\Lambda(\rho)=\sum_{k}p_kU_k\rho U_{k}^{\dag}=\sum_{k}p_k\sigma_k\;,
$$
where $U_k\equiv U_{A_{1}}^{(k)}\otimes U_{A_{2}}^{(k)}\otimes\cdots\otimes U_{A_{n}}^{(k)}$ is a product of unitaries.
Hence, all the states $\sigma_k$ have the same von Neumann entropy $S$ as $\rho$, and we get
$$
S\left(\sum_{k}p_k\sigma_k\right)=S\left(V_{P}\rho V_{P}^{\dag}\right)=S(\rho)=\sum_{k}p_kS(\sigma_k)\;.
$$ 
From the strict concavity of the von Neumann entropy we get that all the $\sigma_k$ are the same and therefore
$V_P=U_{A_{1}}^{(1)}\otimes U_{A_{2}}^{(1)}\otimes\cdots\otimes U_{A_{n}}^{(1)}$ is a product of unitaries.
\end{proof}

In conclusions, Eqs.~(\ref{main},\ref{nway})  provides us, for the first time, with closed expressions for the time
evolution of multipartite entanglement of a composite system interacting locally with the environment. These expressions
emerge from the SL-invariance of the measures defined in Def.~\ref{ginv}, and \emph{not} from the Jamiolkowski isomorphism
which is the methodology used in Ref.~\cite{Kon08,Mar08}. Amazingly, the evolution of multipartite entanglement (under one local channel) is determined completely by the ERF defined in Eq.(\ref{erf}), irrespective to the number of qudits in the system. 
For multi-qubits systems, the ERF has a closed formula given in terms of Wootters concurrence formula.
In other words, there is no need to solve any master equations in order to determine the time evolution of multipartite entanglement.

\emph{Acknowledgments:---}
The author is grateful for numerous illuminating discussions with Nolan Wallach.
This research is supported by NSERC.

\end{document}